\documentclass{llncs}

\usepackage{microtype}
\usepackage{tikz-cd}
\usetikzlibrary{arrows}
\usepackage{multicol}
\usepackage{mathrsfs}
\usepackage{url}
\usepackage{stmaryrd}
\usetikzlibrary{trees}
\usepackage[all]{xy}
\usepackage{adjustbox}
\usepackage{multirow}
\usepackage{algorithm}
\usepackage{algorithmic}
\usepackage{float}
\usepackage{enumerate}
\usepackage{amsmath}
\usepackage{hyperref}
\usepackage{amssymb}
\usepackage[mathscr]{euscript}
\newcommand{\scp}{\hspace{2pt};\hspace{2pt}}
\newcommand\po{\mathscr{P}}
\newcommand\te{\psi}
\newcommand\di{\mathscr{D}}

\newcommand\emo{\mathcal{E\!M}}
\newcommand\res{\mathcal{R}}
\newcommand\kl{\mathcal{K}l}

\renewcommand{\paragraph}[1]{\noindent\textbf{\normalsize #1}}
\newcommand*{\ccld}[1]{{\tikz[baseline=(X.base)]\node(X)[draw,shape=circle,inner sep=0]{\text{\scriptsize\strut$#1$}};}}

\newcommand{\sem}[1]{\llbracket{}#1\rrbracket{}}

\newcommand{\lsg}{\xymatrix{{}\ar@{~>}[r]&{}}}

\newcommand{\lat}{\xymatrix{{}\ar@{-->}[r]&{}}}

\newcommand{\Set}{\mathbf{Set}}
\newcommand{\Mon}{\mathbf{Mon}}
\newcommand{\Alg}[1][\Sigma,E]{\mathbf{Alg}(#1)}
\newcommand{\cat}{\mathbf{C}}

\newcommand{\EM}{\mathcal{E\!M}}

\newcommand{\Pow}{\mathscr{P}}

\newcommand{\poly}[1][\Sigma]{\mathsf{H}_{#1}}
\newcommand{\free}[1][\Sigma]{\mathsf{F}_{#1}}
\newcommand{\forg}[1][\Sigma]{\mathsf{U}_{#1}}
\newcommand{\Forg}{\mathsf{U}}
\newcommand{\Free}{\mathsf{F}}

\newcommand{\str}{\mathsf{st}}
\newcommand{\swp}{\mathsf{swap}}

\newcommand{\N}{\mathbb{N}}
\newcommand{\algb}[1][A]{\mathcal{#1}}

\newcommand{\epi}{\twoheadrightarrow}
\newcommand{\id}{\mathrm{id}}

\newcommand{\prepare}[1][\algb]{\delta^V_{#1}}
\newcommand{\evaluate}[1][\algb]{\sigma^V_{#1}}

\newcommand{\inv}{^{-1}}
\newcommand{\ari}{\mathrm{ar}}
 
\newcommand{\s}{^{\Sigma}}
\newcommand{\lsem}{\llbracket}
\newcommand{\rsem}{\rrbracket}
\newcommand{\prog}[1][p]{\mathtt{#1}}

\usepackage[colorinlistoftodos]{todonotes}
\presetkeys{todonotes}{fancyline, color=blue!30, size=\small}{}

\mathcode`;="003B  
\mathcode`?="003F  
\mathcode`|="026A  
\mathcode`<="4268  
\mathcode`>="5269  

\mathchardef\ls="213C    
\mathchardef\gr="213E    
\mathchardef\uparrow="0222  
\mathchardef\downarrow="0223  

\title{Layer by Layer -- Combining Monads}

\author{Fredrik Dahlqvist \and Louis Parlant \and Alexandra Silva\thanks{This work was partially supported by ERC grant ProfoundNet.}}
 \institute{University College London}

\begin{document}
\maketitle
%
%

\begin{abstract} 
We develop a modular method to build algebraic structures. Our approach is categorical: we describe the layers of our construct as monads, and combine them using distributive laws.

Finding such laws is known to be difficult and our method identifies precise sufficient conditions for two monads to distribute. We either (i) concretely build a distributive law which then provides a monad structure to the composition of layers, or (ii) pinpoint the algebraic obstacles to the existence of a distributive law and suggest a weakening of one layer that ensures distributivity.

This method can be applied to a step-by-step construction of a  programming language. Our running example will involve three layers: a basic imperative language enriched first by adding non-determinism and then probabilistic choice. The first extension works seamlessly, but the second encounters an obstacle, resulting in an `approximate' language very similar to the probabilistic network specification language ProbNetKAT.
\end{abstract}

\section{Introduction}\label{sec:intro}

The practical objective of this paper is to provide a systematic and modular understanding of the design of recent programming languages such as NetKAT~\cite{foster2015coalgebraic} and ProbNetKAT~\cite{foster2016probabilistic,cms} by re-interpreting their syntax as a layering of monads. However, in order to solve this problem, we develop a very general technique for building \emph{distributive laws between monads} whose applicability goes far beyond understanding the design of languages in the NetKAT family. Indeed, the combination of monads has been an important area of research in theoretical computer science ever since Moggi developed a systematic understanding of computational effects as monads in \cite{moggi1991notions}. In this paradigm -- further developed by Plotkin, Power and others in e.g. \cite{plotkin2002notions,benton2002monads} -- the question of how to combine computational effects can be treated systematically by studying the possible ways of combining monads. This work can also be understood as a contribution to this area of research.

 

Combining effects is in general a non-trivial issue, but diverse methods have been studied in the literature. A \emph{monad transformer}, as described in \cite{benton2002monads}, is a way to enrich any theory with a specific effect. These transformers allow a step-by-step construction of computational structures, later exploited by Hudak et al. \cite{liang1995monad,liang1996modular}. 
In \cite{hyland2006combining}, Hyland, Plotkin and Power systematized the study of effect combinations by introducing two canonical constructions for combining monads, which in some sense lie at the extreme ends of the collection of possible combination procedures. At one end of the spectrum they define the \emph{sum} of monads which consists in the juxtaposition of both theories with no interaction whatsoever between computational effects. At the other end of the spectrum they define the \emph{tensor} of two monads where both theories are maximally interacting in the sense that ``each operator of one theory commutes with each operation of the other'' (\cite{hyland2006combining}). In \cite{hyland2004combining} they combine exceptions, side-effects, interactive input/output, non-determinism and continuations using these operations.

In some situations neither the sum nor the tensor of monads is the appropriate construction, and some intermediate level of interaction is required. 
From the perspective of understanding the design of recent programming languages which use layers of non-determinism and probabilities (e.g.~ProbNetKAT), there are two reasons to consider combinations other than the sum or the tensor. First, there is the unavoidable mathematical obstacle which arises when combining sequential composition with non-deterministic choice (see the simple example below), two essential features of languages in the NetKAT family. When combining two monoid operations with the tensor construction, one enforces the equation $\tt (p;q)+(r;s)=(p+r);(q+s)$ which means, by the Eckmann-Hilton argument, that the two operations collapse into a single commutative operation; clearly not the intended construction. Secondly, and much more importantly, the intended \emph{semantics} of a language may force us to consider specific and limited interactions between its operations. This is the case for languages in the NetKAT family, where the intended trace semantics suggests \emph{distributive laws} between operations, for instance that sequential composition distributes over non-deterministic choice (but not the converse). For this reason, the focus of this paper will be to explicitly construct \emph{distributive laws} between monads.



It is worth noting that existence of distributive laws is a subtle question and having automatic tools to derive these is crucial in avoiding mistakes. As a simple example in which several mistakes have appeared in the literature, consider the composition of the  powerset monad $\po$ with itself. Distributive laws of $\po$ over $\po$ were proposed in 1993 by King \cite{king1993combining} and in 2007 \cite{manes2007monad}, with a subsequent correction of the latter result by Manes and Mulry themselves in a follow-up paper. In 2015, Klin and Rot made a similar claim \cite{klin2015coalgebraic}, but recently Klin and Salamanca have in fact showed that there is no distributive law of $\po$ over itself and explain carefully why all the mistakes in the previous results were so subtle and hard to spot \cite{KlinPP}. This example shows that this question is very technical and sometimes counter-intuitive. Our general and modular approach 
provides a fine-grained method for determining \textbf{(a)} if a monad combination by distributive law is possible, \textbf{(b)} if it is not possible, exactly which features are broken by the extension and \textbf{(c)} suggests a way to fix the composition by modifying one of our monads. In other words, this enables informed design choices on which features we may accept to lose in order to achieve greater expressive power in a language through monad composition. 
 
The original motivation for this work is very concrete and came from trying to understand the design of ProbNetKAT, a recently introduced programming language with non-determinism and probabilities ~\cite{foster2016probabilistic,cms}. The non-existence of a distributive law between the powerset monad and the distribution monad, first proved by Varacca \cite{2003:VaraccaProbability} and discussed recently in \cite{bonchi2017power}, is a well known problem in semantics. As we will show, our method enables us to modularly build ProbNetKAT based on the composition of several monads capturing the desired algebraic features. The method derives automatically which equations have to be dropped when adding the probabilistic layer providing a principled justification to the work initially presented in~\cite{foster2016probabilistic,cms}.

\paragraph{A simple example.} Let us consider a set \texttt{P} of atomic programs, and build a `minimal' programming language as follows. Since sequential composition is essential to any imperative language we start by defining the syntax as:
\begin{equation}\label{syn1}
\tt
p::= skip \mid p\scp p \mid a\in P 
\end{equation}
and ask that the following programs be identified:
\begin{equation}\label{law1}
\tt
p\scp skip =p= skip\scp p   \qquad \text{ and } \qquad p\scp(q\scp r)= (p\scp q)\scp r 
\end{equation}
The language defined by the operations of \eqref{syn1} and the equations of \eqref{law1} can equally be described as the application of the \emph{free monoid monad} $(-)^*$ to the set of atomic programs $\tt P$. If we assign a semantics to each basic program $\tt P$,  the semantics of the extended language can be defined as finite sequences (or traces) of the basic semantics. 
In a next step, we might want to enrich this basic language by adding a non-deterministic choice operation $+$ and the constant program $\tt abort$, satisfying the equations:
\begin{equation}\label{law2}
\tt
abort + p = p = p + abort \qquad p+p = p \qquad p+q = q+p \qquad p+(q+r) = (p+q)+r 
\end{equation}
The signature $(\tt abort, +)$ and the axioms \eqref{law2} define join-semilattices, and the monad building free semilattices is the \emph{finitary  powerset monad} $\Pow$. To build our language in a modular fashion we thus want to apply $\Pow$ on top of our previous construction and consider the programming language where the syntax and semantics arise from $\Pow(\tt P^*)$.
For this purpose we combine both monads to construct a new monad $\Pow(-^*)$ by building a distributive law $(-)^*\Pow\to \Pow(-)^*$. As explained above, this approach is semantically justified by the intended trace semantics of the language, and will ensure that operations from the inner layer distribute over the outer ones, i.e. 
\begin{equation}\label{law3}
\tt
p ; (q + r) = p;q + p;r \qquad (q + r) ; p = q;p + r;p \qquad p ; abort = abort ; p = abort 
\end{equation}
Our method proves and relies on the following theorem: if $\Pow$ preserves the structure of $(-)^*$-algebra defined by \eqref{syn1}-\eqref{law1}, then the composition $\Pow(-^*)$ has a monad structure provided by the corresponding distributive law. 
Applying this theorem to our running example, the first step is to lift the signature \eqref{syn1}, in other words to define new canonical interpretations in $\Pow(\tt P^*)$ for $;$ and $\tt skip$. Once this lifting is achieved, the equations in \eqref{law1}, arising from the inner layer, can be interpreted in $\Pow(-^*)$. We need to check if they still hold: is the new interpretation of $;$ still associative? To answer this question, our method makes use of categorical diagrams to obtain precise conditions on our monadic constructs. Furthermore, in the case where equations fail to hold, we provide a way to identify exactly what stands in the way of monad composition. We can then offer tailor-made adjustments to achieve the composition and obtain a `best approximate' language, with slightly modified monads.

\paragraph{Structure of this paper.} Section \ref{sec:primer} presents some basic facts about monads and distributive laws and fixes the notation. In Section \ref{sec:distlaw1} we recall the well-known fact \cite{sokolova2007generic,milius2009complete} that there exists a distributive law of any polynomial functor over a monoidal $\Set$-monad. In particular this shows that \emph{operations} can be lifted by monoidal monads. In fact, the techniques presented in this paper can be extended beyond the monoidal case, but since we won't need such monads in our applications, we will focus on monoidal monads for which the lifting of operations is very straightforward.  We then show in Section \ref{sec:diagres} when \emph{equations} can also be lifted. We isolate two conditions on the lifting monad which guarantee that any equation can be lifted. These two conditions correspond to a monad being \emph{affine} \cite{kock1972bilinearity} and \emph{relevant} \cite{jacobs1994semantics}. We also characterise the general form of equations preserved by monads which only satisfy a subset of these conditions. Interestingly, together with the symmetry condition \eqref{diag:monoidal:sym} which is always satisfied by monoidal $\Set$-monads, we recover what are essentially the three structural laws of classical logic (see also \cite{jacobs1994semantics}). In Section \ref{sec:application} we show how the $\ast$-free fragment of ProbNetKAT can be built in systematic way by construction distributive laws between the three layers of the language. 

\section{A primer on monads, algebras and distributive laws}\label{sec:primer}
\paragraph{\bf Monads and  $(\Sigma,E)$-algebras.}  
For the purposes of this paper, we will always consider monads on $\Set$ (\cite{awodey,mac2013categories,moggi1991notions}). 
The core language described in the introduction is defined by the signature $\Sigma=\{\scp,\tt skip\}$ and the set $E$ of equations given by \eqref{law1}. More generally, we view programming languages as algebraic structures defined by a signature $(\Sigma, \ari:\Sigma\to\N)$ and a set of equations $E$ enforcing program equivalence. To formalize this we first define a $\Sigma$-algebra to be a set $X$ together with an interpretation $\lsem \sigma\rsem  \colon X^{\ari(\sigma)} \to X$ of each operation $\sigma\in\Sigma$. A $\Sigma$-algebra can be conveniently represented as an algebra for the polynomial functor $\poly=\coprod_{\sigma\in\Sigma}(-)^{\ari(\sigma)}$ defined by the signature, i.e. as a set $X$ together with a map $\beta\colon \poly X\to X$. A $\Sigma$-algebra morphism between $\beta\colon \poly X\to X$ and $\gamma\colon \poly Y\to Y$ is a map $f\colon X\to Y$ such that $\gamma\circ \poly f=f\circ \beta$. The category of $\Sigma$-algebras and $\Sigma$-algebra morphisms is denoted $\Alg[\Sigma]$.
In particular, the set $\free X$ of all $\Sigma$-terms is a $\Sigma$-algebra -- the free $\Sigma$-algebra over $X$ -- and $\free$ is a functor $\Set\to\Alg[\Sigma]$ forming an adjunction
\begin{equation}\label{diag:adjunctionR}
\free \dashv \Forg_\Sigma:\Alg[\Sigma]\to \Set
\end{equation} 
Since it will not lead to any ambiguity we will usually overload the symbol $\free$ to also denote the monad $\Forg_\Sigma \free:\Set\to\Set$ arising from this adjunction.

Given a $\Sigma$-algebra $\mathcal{A}$, a free $\Sigma$-term $s$ built over variables in a set $V$, and a valuation map $v: V\to \Forg_\Sigma \mathcal{A}$, we define the interpretation $\lsem s\rsem_v$ of $s$ in $\mathcal{A}$ recursively in the obvious manner. We say that an equation $s=t$ between free $\Sigma$-terms is valid in $\mathcal{A}$, denoted $\mathcal{A}\models s=t$, if for every valuation $v: V\to \Forg_\Sigma \mathcal{A}$, $\lsem s\rsem_v=\lsem t\rsem_v$. Given a set $E$ of equations we define a $(\Sigma,E)$-algebra as a $\Sigma$-algebra in which all the equations in $E$ are valid. We denote by $\Alg$ the subcategory of $\Alg[\Sigma]$ consisting of $(\Sigma,E)$-algebras. There exists a functor $\Free:\Set\to\Alg$ building free $(\Sigma,E)$-algebras which is left adjoint to the obvious forgetful functor:
\begin{equation}\label{diag:adjunction}
\Free \dashv \Forg:  \Alg[\Sigma, E]\to\Set
\end{equation} 
In our running example all monads arise from a finitary syntax, and thus from an adjunction of the type \eqref{diag:adjunction}.
%

\paragraph{\bf  Eilenberg-Moore categories.} An algebra for the monad $T$ is a set $X$ together with an map $\alpha\colon TX\to X$ such that the diagrams in \eqref{algebras} commute. A morphism $ (X,\alpha)  \stackrel{f}{\rightarrow}  (Y,\beta)$ of $T$-algebras is a morphism  $ X \stackrel{f}{\rightarrow} Y $ in $\Set$ verifying $\beta \circ Tf =f \circ \alpha$.
\vspace{-.6cm}\begin{align}
\xymatrix@R=.35cm{
 	{TTX}  	
 	 				\ar[r]^{\mu_X} 
 	 				\ar[d]_{T \alpha}
 & 				
 	{TX}  	
 	 				\ar[d]^{\alpha} 	 				
 & & X 			
 					\ar[r]^{\eta_X} 
 					\ar[rd]_{1}
 & TX 
 					\ar[d]^{\alpha}
 \\
  	{TX}  	
 	 				\ar[r]^{\alpha} 
 & 				
 	{X}  
 & & &	{X}  
}\label{algebras}
\end{align}
The category of $T$-algebras and $T$-algebra morphisms is called the \emph{Eilenberg-Moore} category of the monad $T$, and denoted $\emo(T)$. There is an obvious forgetful functor $\Forg_E:\emo(T)\to \Set$ which sends an algebra to its carrier, it has a left adjoint $\Free_E:\Set\to\emo(T)$ which sends a set $X$ to the free $T$-algebra $\mu_X: T^2 X\to TX$. Note that the adjunction $\Free_E\dashv \Forg_E$ gives rise to the monad $T$. 
 A \emph{lifting} of a functor $F: \Set \to \Set$ to $\emo(T)$ is a functor $\hat{F}$ on  $\emo(T)$ such that $\Forg_E\circ\widehat{F}=F\circ \Forg_E$
\begin{lemma}[\cite{mac2013categories} VI.8. Theorem 1]\label{lem:EM_Alg_Equiv}
For any adjunction of the form \eqref{diag:adjunction},  $\EM(\Forg\Free)$ and $\Alg$ are equivalent categories.
\end{lemma}
The functors connecting $\EM(\Forg\Free)$ and  $\Alg$ are traditionally called \emph{comparison functors}, and we will denote them by $M\colon \EM(\Forg\Free)\to\Alg$ and $K\colon \Alg\to\EM(\Forg\Free)$.
Consider first the free monad $\free$ for a signature $\Sigma$ (i.e. the monad generated by the adjunction \eqref{diag:adjunctionR}). The comparison functor $M\colon \Alg[\Sigma]\to\EM(\free)$ maps the free $\free$-algebra over $X$, that is $\mu^{\free}_X\colon \free^2 X\to \free X$ to the free $\poly$-algebra over $X$ which we shall denote by $\alpha_X: \poly \free X\to \free X$. It is well-known that $\alpha_X$ is an isomorphism. Moreover, the maps $\alpha_X$ define a natural transformation $\poly \free\to \free$. Similarly, in the presence of equations, if we consider the adjunction $\Free\dashv \Forg$ of \eqref{diag:adjunction} and the associated monad $T=\Forg\Free$, then the comparison functor $M'\colon \Alg\to\EM(T)$ sends the free $T$-algebra $\mu^T_X: T^2X\to TX$ to an $\poly$-algebra which we shall denote $\rho_X: \poly TX \to TX$. Again, the maps $\rho_X$ define a natural transformation $\poly T\to T$, but in general $\rho_X$ is no longer an isomorphism: in the case of monoids and of a set $X=\{x,y,z\}$, we have $\rho_X(x;(y;z))=\rho_X((x;y);z)$.



\paragraph{\bf  Distributive laws.} Let  $(S, \eta^S, \mu^S)$ and  $(T, \eta^T, \mu^T)$ be monads, a \emph{distributive law of $S$ over $T$}  (see \cite{beck1969distributive}) is a natural transformation $\lambda: ST\to TS$ satisfying:

\begin{adjustbox}{max width=\textwidth}
\begin{tabular}{l l l l}
\begin{minipage}[t]{3cm}
\begin{equation}
\xymatrix@C=4ex
{
& S\ar[dl]_{S\eta^T}\ar[dr]^{\eta^T S}\\
ST\ar[rr]_{\lambda} & & TS
}\label{diag:distlaw:etaT}\tag{DL. 1}
\end{equation}
\end{minipage}

&

\begin{minipage}[t]{3cm}
\begin{equation}
\xymatrix@C=4ex
{
& T\ar[dl]_{\eta^S T}\ar[dr]^{T\eta^S}\\
ST\ar[rr]_{\lambda} & & TS
}\label{diag:distlaw:etaS}\tag{DL. 2}
\end{equation}
\end{minipage}

&

\begin{minipage}[t]{5cm}
\begin{equation}
\xymatrix@C=4ex
{
STT\ar[d]_{S \mu^T}\ar[r]^{\lambda T} & TST\ar[r]^{T\lambda} & TTS\ar[d]_{\mu^T S}\\
ST\ar[rr]_{\lambda} & & TS
}\label{diag:distlaw:muT}\tag{DL. 3}
\end{equation}
\end{minipage}

&

\begin{minipage}[t]{4cm}
\begin{equation}
\xymatrix@C=4ex
{
SST\ar[d]^{\mu^S T}\ar[r]^{S \lambda} & STS\ar[r]^{\lambda S} & TSS\ar[d]^{T\mu^S}\\
ST\ar[rr]_{\lambda} & & TS
}\label{diag:distlaw:muS}\tag{DL. 4}
\end{equation}
\end{minipage}
\end{tabular}
\end{adjustbox}
\vspace{2ex}

\noindent If $\lambda$ only satisfies \eqref{diag:distlaw:etaS} and \eqref{diag:distlaw:muS}, we will say that $\lambda$ is a distributive law of the the monad $S$ over \emph{functor} $T$, or in the terminology of \cite{trs}, an $\emo$-law of $S$ over $T$. Dually, if $\lambda$ only satisfies \eqref{diag:distlaw:etaT} and \eqref{diag:distlaw:muT}, $\lambda$ is known as a distributive law of the \emph{functor} $S$ over the monad  $T$, or $\kl$-law of $S$ over $T$ \cite{trs}.

\begin{theorem}{\cite{beck1969distributive,trs,balan2011coalgebras}}\label{thm:distLaws_Liftings}
$\emo$-laws $\lambda: SF\to FS$ and liftings of $F$ to $\emo(S)$ are in one-to-one correspondence.
\end{theorem}
If there exists a distributive law $\lambda: TS\to ST$ of the monad $T$ over the monad $S$, then the composition of $S$ and $T$ also forms a monad $(ST,u, m)$, whose unit $u$ and multiplication $m$ are given by:

\begin{adjustbox}{width=\textwidth}
\xymatrix{
 	{X}  	
 	 				\ar[r]^{ \eta_{X}^T} 
 	 				\ar@/_1pc/[rr]|{\;u_X\;} 
 & 				
 	{TX}  	
 	 				\ar[r]^-{ \eta_{TX}^S}
& 
	{STX}
	&
	{STST X}  	
 	 				\ar[r]^{ S \lambda_{TX}} 
 	 				\ar@/_1pc/[rrr]|{\;m_X\;} 
 & 				
 	{SSTTX}  	
 	 				\ar[r]^{ \mu_{TTX}^S}
&
 	{STTX}  	
 	 				\ar[r]^{ S \mu_{X}^T}
&
 	{STX} 
}
\end{adjustbox}

If $\EM(S)\simeq \Alg$ and $\EM(T)\simeq \Alg[\Sigma',E']$, then a distributive law $ST\to TS$ implements the distributivity of the operations in $\Sigma$ over those of $\Sigma'$.
%
%
%

\section{Building distributive laws between monads}
\label{sec:distlaw1}

In this section we will show how to construct a distributive law $\lambda\colon ST\to TS$  between monads via a \emph{monoidal structure} on $T$. 

\subsection{Monoidal monads}
Let us briefly recall some relatively well-known categorical notion.  A \emph{lax monoidal functor} on a monoidal category $(\cat,\otimes,I)$, or simply a monoidal functor\footnote{We will never consider the notion of \emph{strong} monoidal functor, so this terminology should not lead to any confusion.}, is an endofunctor $F:\cat\to\cat$ together with natural transformations $\psi_{X,Y}: FX\otimes FY\to F(X\otimes Y)$ and $\psi^0: I\to FI$ satisfying the diagrams:

\vspace{-20pt}
\begin{adjustbox}{max width=\textwidth}
\hspace{-0.5cm}\begin{tabular}{l  c  r}
\begin{minipage}[t]{4cm}
\begin{equation}
\xymatrix@C=7ex
{
FX\otimes I\ar[r]^{\id_{FX}\otimes \psi^0}\ar[r]^{\id_{FX}\otimes \psi^0}\ar[d]^{\rho_{FX}} & FX\otimes FI\ar[d]_{\psi_{X,I}} \\
FX & F(X\otimes I)\ar[l]_{F\rho_X}
}\label{diag:monoidal:left_unit}\tag{MF. 1}
\end{equation}
\end{minipage}
&
\multirow{2}{*}{ \begin{minipage}[t]{6.5cm}
\vspace{1cm}
\begin{equation}
\xymatrix@C=12ex
{
(FX\otimes FY)\otimes FZ\ar[r]^{\alpha_{FX,FY,FZ}} \ar[d]^{\psi_{X,Y}\otimes \id_{FZ}}& FX\otimes (FY\otimes FZ) \ar[d]_{\id_{FX}\otimes \psi_{Y,Z}}\\
F(X\otimes Y)\otimes FZ\ar[d]_{\psi_{X\otimes Y, Z}} & FX\otimes F(Y\otimes Z)\ar[d]_{\psi_{X,Y\otimes Z}}\\
F((X\otimes Y)\otimes Z)\ar[r]^{F\alpha_{X,Y,Z}} & F(X\otimes (Y\otimes Z)))
}\label{diag:monoidal:assoc}\tag{MF. 3}
\end{equation}
\end{minipage}}
\\
\begin{minipage}[t]{4cm}
\begin{equation}
\xymatrix@C=7ex
{
I\otimes FX\ar[r]^{\psi^0\otimes \id_{FX}}\ar[r]^{\psi^0\otimes \id_{FX}}\ar[d]^{\rho'_{FX}} & FI\otimes FX\ar[d]_{\psi_{I,X}} \\
FX & F(I\otimes X)\ar[l]_{F\rho'_X}
}\label{diag:monoidal:right_unit}\tag{MF. 2}
\end{equation}
\end{minipage}
\end{tabular}
\end{adjustbox}

\noindent where $\alpha$ is the associator of $(\cat,\otimes, I)$ and $\rho,\rho'$ the right and left unitors respectively. The diagrams \eqref{diag:monoidal:left_unit}, \eqref{diag:monoidal:right_unit} and \eqref{diag:monoidal:assoc} play a key role in the lifting of operations and equations in this section and the next. In particular they ensure that any unital (resp. associative) operation lifts to a unital (resp. associative) operation. We will sometimes refer to $\psi$ as the \emph{Fubini transformation} of $F$.

A \emph{monoidal monad} $T$ on a monoidal category is a monad whose underlying functor is monoidal for a natural transformation $\psi_{X,Y}: TX\otimes TY\to T(X\otimes Y)$ and $\psi^0=\eta_I$, the unit of the monad at $I$, and whose unit and multiplication are monoidal natural transformations, that is to say:

\begin{adjustbox}{max width=\textwidth}
\vspace{-7mm}\hspace{-0.5cm}\begin{tabular}{l c r}
\begin{minipage}[t]{5.5cm}
\begin{equation}
\xymatrix
{
X\otimes Y\ar[r]^{\eta_X\otimes \eta_Y} \ar[dr]_{\eta_{X\otimes Y}}& TX\otimes TY\ar[d]^{\psi_{X,Y}}\\
& T(X\otimes Y)
}\label{diag:monoidal:unit}\tag{MM.1}
\end{equation}
\end{minipage}
&
\hspace{4cm}
&
\begin{minipage}[t]{6cm}
\begin{equation}
\mspace{-195mu}\xymatrix
{
T^2X\otimes T^2 Y\ar[d]_{\mu_{X}\otimes \mu_{Y}} \ar[r]^{\psi_{TX,TY}} & T(TX\otimes TY) \ar[r]^{T\psi_{X,Y}}& TT(X\otimes Y)\ar[d]_{\mu_{X\otimes Y}}\\
TX\otimes TY\ar[rr]^{\psi_{X,Y}} & & T(X\otimes Y)
}\label{diag:monoidal:multiplication}\tag{MM.2}
\end{equation}
\end{minipage}
\end{tabular}
\end{adjustbox}

\noindent Moreover, a monoidal monad is called \emph{symmetric monoidal} if
\begin{equation}
\xymatrix
{
TX\otimes TY\ar[r]^{\psi_{X,Y}}\ar[d]_{\swp_{TX,TY}} & T(X\otimes Y)\ar[d]^{T\swp_{X,Y}}\\
TY\otimes TX\ar[r]_{\psi_{Y,X}} & T(Y\otimes X)
}\label{diag:monoidal:sym}\tag{SYM}
\end{equation}
where $\swp: (-)\otimes(-)\to (-)\otimes (-)$ is the argument-swapping transformation (natural in both arguments). 

We now present a result which shows that for monoidal categories which are sufficiently similar to $(\Set,\times ,1)$, being monoidal is equivalent to being symmetric monoidal. The criteria on $(\cat,\otimes ,I)$ in the following theorem are due to \cite{sato2017giry} and generalize the strength unicity result of \cite[Prop. 3.4]{moggi1991notions}. Our usage of the concept of strength in what follows is purely technical, it is the monoidal structure which is our main object of interest. We therefore refer the reader to e.g. \cite{moggi1991notions} for the definitions of strength and commutative monad.

\begin{theorem}\label{thm:symMon}
Let $T:\cat\to\cat$ be a monad over a monoidal category $(\cat,\otimes, I)$ whose tensor unit $I$ is a separator of $\cat$ (i.e.  $f,g: X\to Y$ and $f\neq g$ implies $\exists x: I\to X$ s.th. $f\circ x\neq g\circ x$)  and such that for any morphism $z:I\to X\otimes Y$ there exist $x: I\to X, y\to Y$ such that $z=(x\otimes y)\circ \rho_I\inv$. Then t.f.a.e.
\begin{enumerate}[(i)]
\item There exists a unique natural transformation $\psi_{X,Y}: TX\otimes TY\to T(X\otimes Y)$ making $T$ monoidal
\item There exists a unique strength $\str_{X,Y}: X\times TY\to T(X\otimes Y)$ making $T$ commutative
\item There exists a unique natural transformation $\psi_{X,Y}: TX\otimes TY\to T(X\otimes Y)$ making $T$ symmetric monoidal
\end{enumerate}
\end{theorem}

In particular, monoidal monads on $(\Set,\times,1)$ are necessarily symmetric (and thus commutative). As we will see in the next section (Theorem \ref{thm:preservation:noConditions}), this symmetry has deep consequences: it means that a large syntactically definable class of equations can always be lifted by monoidal monads. 
\subsection{Lifting operations}
First though, we show that being monoidal allows us to lift \emph{operations}. The following Theorem is well-known and can be found in e.g. \cite{sokolova2007generic,milius2009complete}.
\begin{theorem}\label{thm:monoidalDLaw}
Let $T:\Set\to\Set$ be a monoidal monad, then  for any finitary signature $\Sigma$, there exists a distributive law $\lambda^\Sigma: \poly T\to T \poly$ of the polynomial functor associated with $\Sigma$ over $T$.
\end{theorem}

The distributive laws $\lambda\s: \poly T\to T\poly$ built from a monoidal structure $\psi$ on $T$ in Theorem \ref{thm:monoidalDLaw} have the general shape
\begin{equation}\label{diag:HSoverT}
\xymatrix@C=14ex
{
\poly T X=\coprod_{s\in\Sigma} (TX)^{\ari(s)} \ar[r]^-{\coprod_{s\in\Sigma}\psi^{(\ari(s))}_X}  & T \poly X
}
\end{equation}
where $\psi^{(0)}_X=\eta_1^T, \psi^{(1)}_X=\id_X, \psi^{(2)}_X=\psi_{X,X}$. For $k \geq 3$ if we wanted to be completely rigorous we should first give an evaluation order to the $k$-fold monoidal product $(TX)^k$ -- for example evaluating the products from the left, e.g. $(TX)^3:= (TX\otimes TX)\otimes TX$ -- and then define $\psi^{(k)}: (TX)^k\to T(X^k)$ accordingly by repeated application of the Fubini transformation $\psi$ -- for example defining 
\[
\psi^{(3)}_X=\psi_{X\otimes X,X}\circ (\psi_{X,X}\times \id): (TX\otimes TX)\otimes TX \to T((X\otimes X)\otimes X)
\]
However, we will in general be interested in a variety of evaluation orders for the tensors (depending on circumstances), and since in $\Set$ these different evaluation orders are related by a combination of associators $\alpha_{X,Y,Z}$ which simply re-bracket tuples, we will abuse notation slightly and write
\[
\psi^{(k)}_X: (TX)^k\to T(X^k)
\]
with the understanding that $\psi^{(k)}_X$ is only defined up to re-bracketing of tuples which is quietly taking place `under the hood' as called for by the particular situation. The distributive laws defined by Theorem \ref{thm:monoidalDLaw} can be extended to distributive laws for the free monad associated with the signature $\Sigma$.
\begin{proposition}\label{prop:distributionSigma}
Given a finitary signature $\Sigma$ and a monad $T:\Set\to\Set$, there is a one-to-one correspondence between
\begin{enumerate}[(i)]
\item  distributive laws $\lambda^\Sigma: \poly T\to T \poly$ of the polynomial \emph{functor} associated with $\Sigma$ over $T$
\item distributive laws $\rho^\Sigma: \free T\to T \free$ of the free \emph{monad} associated with $\Sigma$ over $T$
\end{enumerate}
\end{proposition}
In particular, by Theorem \ref{thm:distLaws_Liftings}, the distributive law \eqref{diag:HSoverT} also corresponds to a lifting $\widehat{T}$ of $T$ to $\emo(\free)\simeq \Alg[\Sigma]$. Explicitly, given an $\free$-algebra $\beta: \free X\to X$, $\widehat{T}(X,\beta)$ is defined as the $\free$-algebra 
\begin{equation}\label{diag:liftDef}
\xymatrix
{
\free TX\ar[r]^{\rho\s_X} & T\free X\ar[r]^{T\beta} & TX
}
\end{equation}
Thus whenever $T$ is monoidal, we can `lift' the operations of $\Sigma$, or, in programming language terms, we can define the operations of the outer layer ($T$) on the language defined by the operations of the inner layer ($\free$).
\subsection{Lifting equations}
We now show how to go from a lifting of $T$ on $\emo(\free)\simeq\Alg[\Sigma]$ to a lifting of $T$ on $\emo(S)\simeq\Alg$. More precisely, we will now show how to `quotient' the distributive law $\rho\s\colon \free T\to T\free $ into a distributive law $\lambda: ST\to TS$. Of course this is not always possible, but in the next section we will give sufficient conditions under which the procedure described below does work.  The first step is to define the natural transformation $q\colon \free \epi S$
which quotients the free $\Sigma$-algebras by the equations of $E$ to build the free $(\Sigma,E)$-algebra. At each set $X$, let $EX$ denote the set of pairs $(s,t)\in \free X$ such that $SX\models s=t$ and let $\pi_1,\pi_2$ be the obvious projections. Then $q$ can be constructed via the coequalizers:
\begin{equation}\label{diag:qDef}
\xymatrix
{
EX \ar@<3pt>[r]^{\pi_1} \ar@<-3pt>[r]_{\pi_2} & \free X\ar@{->>}[r]^{q_{X}}  & SX
}
\end{equation}
By construction $q$ is a component-wise regular epi monad morphism ($q\circ \eta=\eta^S$ and $\mu^S\circ qq=q\circ \mu^T)$, and it induces a functor $Q:\EM(S)\to \EM(\free)$ defined by 
\[
Q(\xi\colon SX\to X)=\xi\circ q_X: \free X\to X, \hspace{3em}Q(f)=f
\]
which is well defined by naturality of $q$. This functor describes an embedding, in particular it is injective on objects: if $Q(\xi_1)=Q(\xi_2)$ then $\xi_1\circ q_X=\xi_2\circ q_X$, and therefore $\xi_1=\xi_2$ since $q_X$ is a (regular) epi. 

Given two terms $u,v\in \free V$, we will say that a lifting $\widehat{T}: \Alg[\Sigma]\to\Alg[\Sigma]$ preserves the equation $u=v$, or by a slight abuse of notation that the monad $T$ preserves $u=v$, if $\widehat{T}\mathcal{A}\models u=v$ whenever $\mathcal{A}\models u=v$. Similarly, we will say that $\widehat{T}$ sends $(\Sigma,E)$-algebras to $(\Sigma,E)$-algebras if it preserves all the equations in $E$.  Half of the following result can be found in \cite{bonsangue2013presenting} where a distributive law over a \emph{functor} is built in a similar way.

\begin{lemma}\label{thm:quotient}
If $q\colon \free\epi T$ is a component-wise epi monad morphism, $\rho\s$ is a distributive law of the monad $\free$ over the monad $T$ and if there exists a natural transformation $\lambda\colon ST\to TS$ such that the following diagram commutes
\begin{equation}\label{diag:quotient}
\xymatrix@R=4ex
{
\free T\ar@{->>}[r]^{qT}\ar[d]_{\rho\s} & ST\ar[d]^{\lambda}\\
T\free \ar[r]_{Tq} & TS
}
\end{equation} 
then $\lambda$ is a distributive law of the monad $S$ over the monad $T$.
\end{lemma}

From this lemma we can give an abstract criterion which, when implemented concretely in the next section, will allow us to go from a lifting of $T$ on $\emo(\free)\simeq\Alg[\Sigma]$ to a lifting of $T$ on $\emo(S)\simeq\Alg$.

\begin{theorem}\label{thm:puttingEverythingTogether}
Suppose $T,S:\Set\to\Set$ are finitary monads, that $T$ is monoidal and that $\emo(S)\simeq\Alg$, and let $\widehat{T}:\Alg[\Sigma]\to\Alg[\Sigma]$ be the unique lifting of $T$ defined via Theorems \ref{thm:distLaws_Liftings},\ref{thm:monoidalDLaw} and Proposition \ref{prop:distributionSigma}. If $\widehat{T}$ sends $(\Sigma,E)$-algebras to $(\Sigma,E)$-algebras, then there exists a natural transformation $\lambda: ST \to TS$ satisfying \eqref{diag:quotient}, and therefore a distributive law of $S$ over $T$.
\end{theorem}

%


\section{Checking equation preservation}
\label{sec:diagres}
In Section \ref{sec:distlaw1} we showed how to build a lifting of $T\colon\Set\to\Set$ to $\widehat{T}:\Alg[\Sigma]\to\Alg[\Sigma]$ using a Fubini transformation $\te$ via \eqref{diag:HSoverT} and \eqref{diag:liftDef}. In this section we provide a sound method to ascertain whether this lifting sends $(\Sigma,E)$-algebras to $(\Sigma,E)$-algebras, by giving sufficient conditions for the preservation of equations. We assume throughout this section that $T$ is monoidal, in particular $T$ lifts to $\Alg[\Sigma]$ for any finitary signature $\Sigma$. We will denote by $\forg:\Alg[\Sigma]\to\Set$ the obvious forgetful functor.

\subsection{Residual diagrams}
We fix a finitary signature $\Sigma$ and let $u,v$ be $\Sigma$-terms over a set of variables $V$. Recall that the monad $T$ preserves the equation $u=v$ if $\widehat{T}\algb\models u=v$ whenever $\algb\models u=v$. If $t$ is a $\Sigma$-term, we will denote by $Var(t)$ the \emph{set of variables} in $t$ and by $Arg(t)$  the \emph{list of arguments} used in $t$ ordered as they appear in $t$. For example, the list of arguments of $t=f(x_1,g(x_3,x_2),x_1)$ is $Arg(t)=[x_1,x_3,x_2,x_1]$.  

Let $V$ be a set of variables and $\mathcal{A}$ be a $\Sigma$-algebra with carrier $A$, we define the morphism $\prepare(t): A^{|V|} \to A^k$ where $k=|Arg(t)|$ as the following pairing of projections:
\[
\text{if } Arg(t)=[x_{i_1},x_{i_2},x_{i_3}, \dots x_{i_k}] \text{ then } \prepare(t) = \langle \pi_{i_1},\pi_{i_2},\pi_{i_3}, \dots \pi_{i_k}\rangle
\] 
Intuitively, this pairing rearranges, copies and duplicates the variables used in $t$ to match the arguments. Next, we define $\evaluate(t) \colon A^{k} \to A$ inductively by:
\begin{align*}
	\evaluate(x) &
		= \id_A \\
	\evaluate(f(t_1, \ldots, t_{i})) &
		= A^{k} \xrightarrow{\evaluate(t_1) \times \ldots \times \evaluate(t_i)} A^{i} \xrightarrow{f_\mathcal{A}} A
\end{align*}
With $f_\mathcal{A}$ the interpretation of $f \in \Sigma$ in $\mathcal{A}$. Finally we define $\sem{t}_{\algb}^V$~as~$\evaluate(t) \circ \prepare(t)$. The following lemma follows easily from the definitions.
\begin{lemma}
For any $t\in\free V$, $\prepare(t), \evaluate(t)$, and thus $\sem{t}_{\algb}^V$, are natural in $\mathcal{A}$.
\end{lemma}
We can therefore re-interpret any term $t\in \free V$ as a natural transformation $\sem{t}^V: (-)^{(|V|)}\forg \to \forg$ which is itself the composition of two natural transformations. The first one, $\prepare[](t): (-)^{|V|}\forg \to (-)^k\forg $, `prepares' the variables by swapping, copying and deleting them as appropriate. The second one, $\evaluate[](t):(-)^k\forg \to \forg$, performs the evaluation at each given algebra. Of course, the usual soundness and completeness property of term functions still holds.
\begin{lemma}\label{lem:soundcompl}
For $\algb$ a $\Sigma$-algebra and $u,v\in\free V$, $\sem{u}_{\algb}^V = \sem{v}_{\algb}^V$ iff $\algb\models u=v$.
\end{lemma}
Now consider the following diagram:
\begin{equation}\label{diag:preserves}
\xymatrix@C=12ex
{
(-)^{(|V|)}\forg \widehat{T}\ar@/^2pc/[rr]^{\sem{t}^V_{\widehat{T}}} \ar[r]^{\prepare[\widehat{T}](t)}\ar[d]_{\te^{|V|}_{\forg}}\ar@{}[rd]|{\ccld{r}} & 
(-)^k\forg\widehat{T} \ar@{}[rd]|{\ccld{q}} \ar[r]^{\evaluate[\widehat{T}](t)}\ar[d]^{\te^{(k)}_{\forg}} & 
\forg\widehat{T}\ar[d]^{\forg \id_{\widehat{T}}}\\
T(-)^{|V|}\forg \ar[r]_{T \prepare[](t)}\ar@/_2pc/[rr]_{T\sem{t}^V} & 
T(-)^k\forg \ar[r]_{T \evaluate[](t)} & 
\forg\widehat{T}
}
\end{equation}
Since $\forg \circ \widehat{T}=T\circ\forg$  by definition of liftings it is clear that the vertical arrows $\psi_{\forg}^{(|V|)}$ and $\psi_{\forg}^{(k)}$ are well-typed. 
We define $Pres(T,t,V)$ as the outer square of Diagram \eqref{diag:preserves}  and we call the left-hand square $\ccld{r}$ the \emph{residual diagram} $\res(T,t,V)$. The following Lemma is at the heart of our method for building distributive laws.

\begin{lemma}\label{lem:residual}
If $\res(T,t,V)$ commutes, then $Pres(T,t,V)$ commutes.
\end{lemma}

The following soundness theorem follows immediately from Lemma \ref{lem:residual}.
\begin{theorem}\label{thm:residual}
If $u,v\in\free V$ are such that $\res(T,u,V)$ and $\res(T,v,V)$ commute, then $T$ preserves $u=v$.
\label{resi}
\end{theorem}
\begin{proof}
If $\algb\models u=v$, then $\sem{u}^V_{\algb}=\sem{v}^V_{\algb}$ by Lemma \ref{lem:soundcompl} and thus $T\sem{u}^V_{\algb}\circ \psi_{A}^{(|V|)}=T\sem{v}^V_{\algb}\circ \psi_{A}^{(|V|)}$. Since $\res(T,u,V)$ and $\res(T,v,V)$ commute, so do $Pres(T,u,V)$ and $Pres(T,v,V)$ by Lemma \ref{lem:residual}, and therefore $\sem{u}_{\widehat{T}\algb}^V=\sem{v}_{\widehat{T}\algb}^V$, that is to say  $\widehat{T}\algb\models u=v$ by Lemma \ref{lem:soundcompl}.
\end{proof}

Therefore residual diagrams act as sufficient conditions for equation preservation. Note that these diagrams only involve  $\te$, projections and the monad $T$, sometimes inside pairings. In other words, the actual operations of $\Sigma$ appearing in an equation have no impact on its preservation. What matters is the variable rearrangement transformations $\prepare[](u)$ and $\prepare[](v)$, and how they interact with the Fubini transformation $\te$.

The converse of Theorem \ref{thm:residual} does not hold. Consider the powerset monad $\po$ and a $\Sigma$-algebra $\mathcal{A}$ with $\Sigma$ containing a binary operation $\bullet$. Clearly $\widehat{\po}\mathcal{A}\models x \bullet x = x \bullet x$ whenever $\mathcal{A}\models x \bullet x = x \bullet x$, because the equation trivially holds in any $\Sigma$-algebra. In other words, it is preserved by $\po$. However $\res(\po, x \bullet x, \{x\})$ does not commute: provided that $X$ has more than one element, it is easy to see that $\res(\po, x \bullet x, \{x\})$ evaluated at $X$ is
\[
\xymatrix@R=3ex@C=8ex
{
\po A\ar[d]_{\id_{\po A}}\ar[r]^{\Delta_{\po A}} & (\po A)^2\ar[d]^{-\times-}\\
\po A\ar[r]_{\po(\Delta_A)} & \po (A^2)
}
\]
where $\Delta$ is the diagonal transformation and $-\times-$ is the monoidal structure for $\po$ which takes the Cartesian product. This diagram does not commute (in other words $\po$ is not `relevant', see below).


\subsection{Examples of residual diagrams}

We need \emph{a priori} two diagrams per equation to verify preservation. However, in many cases diagrams will be trivially commuting. For instance,  associativity and unit produce trivial diagrams. For associativity we assume a binary operation $\bullet\in\Sigma$, let $V=\{x,y,z\}$ and compute that $\prepare(x \bullet (y \bullet z))=\langle \pi_1,\pi_2,\pi_3 \rangle :~A^3 \to~A^3$ which is just $\id_{A^3}$. It follows that $\res(T,x \bullet (y \bullet z),V)$ commutes since $\te^3 \circ \id_{TA^3} = T\id_{A^3} \circ \te^3$ which trivially holds. The argument for $(x\bullet y)\bullet z$ is identical, thus associativity is \emph{always} lifted. The same argument shows that units are always lifted as well. This is not completely surprising since we have built-in units and associativity via Diagrams \eqref{diag:monoidal:left_unit}, \eqref{diag:monoidal:right_unit} and \eqref{diag:monoidal:assoc}.

Let us now consider commutativity: $x \bullet y = y \bullet x$. In this case, we put $V=\{x,y\}$ and hence $\prepare(x \bullet y) = \id_\mathcal{A}$ and $\res(T,x \bullet y,V)$ obviously commutes for the same reason as before. Similarly, it is not hard to check that $\res(T,y \bullet x,V)$ is just diagram \eqref{diag:monoidal:sym}, which we know holds by our assumption that $T$ is monoidal and Theorem \ref{thm:symMon}. It follows that:
\begin{theorem}
Monoidal monads preserve associativity, unit and commutativity.
\label{thm:auc}
\end{theorem}

\noindent Some equations are not always preserved by commutative  monads, we present here two important examples.
\begin{equation}\label{ida}
\begin{adjustbox}{max width=\textwidth}
\begin{tabular}{p{7cm} p{5cm}}
Idempotency: $x\bullet x=x$
& Absorption: $x\bullet 0=0$\\
$\res(T,x\bullet x,\{x\})$ given by:
&
$\res(T,x\bullet 0,\{x\}) $ given by:
\\
\xymatrix
{
TA  \ar[d]_{T<\pi_1, \pi_1>}  & TA \ar[l]^{\id} \ar[d]^{<\pi_1, \pi_1>}\\
T(A^2) & (TA)^2 \ar[l]^{\te}\\
}
&
\xymatrix
{
TA  \ar[d]_{T!}  & TA \ar[l]^{\id} \ar[d]^{!}\\
T1 & 1 \ar[l]^{\eta_1}\\
}
\end{tabular}
\end{adjustbox}
\end{equation}

These diagrams correspond to classes of monads studied in the literature.
The residual diagram for idempotency can be expressed as the equation $\te_{A,A} \circ \Delta_{TA} = T \Delta_A$, where $\Delta$ is the diagonal operator. A monad $T$ verifying this condition is called \emph{relevant} by Jacobs in \cite{jacobs1994semantics}. Similarly, one easily shows that the commutativity of the absorption diagram is equivalent to the definition of \emph{affine} monads in \cite{kock1972bilinearity,jacobs1994semantics}.

\subsection{General criteria for equation preservation}

As shown in lemma \ref{lem:residual} and Theorem \ref{thm:residual}, the interaction between $T$ and the variable rearrangements operated by $\prepare[]$ can provide a sufficient condition for the preservation of equations. We will focus on three important types of interaction between a monad $T$ and rearrangement operations. First, the residual diagram for commutativity, i.e. Diagram \eqref{diag:monoidal:sym}, which corresponds to saying that `$T$ preserves variable swapping', i.e. that $T$ is commutative/symmetric monoidal, or in logical terms to the exchange rule. As we have seen, this condition \emph{must} be satisfied in order to simply lift operations, so we must take it as a basic assumption. Second, the residual diagram for idempotency (leftmost diagram of \eqref{ida}) which corresponds to `$T$ preserves variable duplications', i.e. that $T$ is \emph{relevant}, or in logical terms to the weakening rule. Finally, the residual diagram for absorption  (rightmost diagram of \eqref{ida}) which corresponds to `$T$ allows to drop variables', i.e. $T$ is \emph{affine}, or in logical terms to the contraction rule. To each of these residual diagrams corresponds a syntactically definable class of equations which are automatically preserved by a monad satisfying the residual diagram.

\begin{theorem}\label{thm:preservation:noConditions}
Let $T$ be a commutative monad. If $Var(u)=Var(v)$ and if variables appear exactly once in $u$ and in $v$, then $T$ preserves $u=v$.
\label{thm:sm}
\end{theorem}

Note that this theorem can be found in \cite{manes2007monad}, where this type of equation is called \emph{linear}. Moreover, $\po$ is within the scope of this result, which generalises one direction of Gautam's theorem (\cite{gautam1957validity}). Let us now present original results by first treating the case where variables may appear several times. 
\begin{theorem}\label{thm:preservation:relevant}
Let $T$ be a commutative relevant monad. If $Var(u)=Var(v)$, then $T$ preserves $u=v$.
\label{smr}
\end{theorem}

Commutative relevant monads seem to preserve many algebraic laws. However, in the case where both sides of the equation do not contain the same variables, for instance $x \bullet 0 = 0$, Theorem \ref{smr} does not apply. Intuitively, the missing piece is the ability to \emph{drop} some of the variables in $V$.

\begin{theorem}\label{thm:preservation:affine}
Let $T$ be a commutative  affine monad. If variables appear at most once in $u$ and in $v$, then $T$ preserves $u=v$.
\label{thm:affinepres}
\end{theorem}

Combining the results of Theorems \ref{thm:preservation:relevant} and \ref{thm:preservation:affine}, one gets a very economical -- if very strong -- criterion for the preservation of \emph{all} equations. 

\begin{theorem}\label{thm:preservation:everything}
Let $T$ be a commutative, relevant and affine monad. For all $u$ and $v$, $T$ preserves $u=v$.
\end{theorem}

Examining the existence of distributive laws between algebraic theories, as well as stating conditions on variable rearrangements, has been studied before in terms of Lawvere Theories (see for instance \cite{cheng2011distributive}). Note that for $T$ commutative monad, being both relevant and affine (sometimes called \emph{cartesian}) is equivalent to preserving products, as seen in \cite{kock1972bilinearity}. This confirms that such a monad $T$ preserves all equations of the underlying algebraic structure, in other words it always has a distributive law with any other monad. This is however a very strong condition. An example of this type of monad is $T(X)=X^Y$ for $Y$ an object of $\Set$.


\subsection{Weakening the inner layer when composition fails.}

In the case where a residual diagram fails to commute, we cannot conclude that the equation lifts from $ \mathcal{A}$ to $\widehat{T} \mathcal{A}$. The non-commutativity of the diagram often provides a counter-example which shows that the equation is in fact not valid in $\widehat{T}  \mathcal{A}$ (this is the case of idempotency and distributivity in the next section).

However, if our aim is to build a structure combining all operations  used to define $T$ and $S$, then our method can provide an answer, since it allows us to identify precisely which equations fail to hold. Let $E'$ be the subset of $E$ containing the equations preserved by $T$. A new monad $S'$ can be derived from signature $\Sigma$ and equations $E'$ using an adjunction of type \eqref{diag:adjunction}. Since $E'$ only contains equations preserved by $T$, by theorem \ref{thm:puttingEverythingTogether} the composition $TS'$ creates a monad, and its algebraic structure contains all the constructs derived from the original signature $\Sigma$, as well as the new symbols arising from $T$.

This method for fixing a faulty monad composition follows the idea of loosening the constraints of the \emph{inner} layer, meaning in this case modifying $S$ to construct a monad resembling $TS$. The best approximate language we obtain has the desired signature, but has lost some of the laws described by $S$. We illustrate this method in the following section.

\section{Application}\label{sec:application}
As sketched in the introduction, our method aims to incrementally build an imperative language: starting with sequential composition, we add a layer providing non-deterministic choice, then a layer for probabilistic choice. 

\paragraph{\bf Adding the non-deterministic layer.}
We start with the simple programming language described in the introduction by the signature \eqref{syn1} and equations \eqref{law1} -- or, equivalently, by the monad $(-)^*$ -- and let $\tt A$ be a set of atomic programs. Our minimal language is thus given by $\tt A^*$. Note that the free monoid is not commutative and thus in our method it cannot be used as an outer layer, it has to constitute the core of the language we build. More generally, our method provides a simple heuristic for compositional language building: always start with the non-commutative monad.


We now add non-determinism via the finitary powerset monad $\Pow$, which is simply the free join semi-lattice monad. To build this extension, we want to combine both monads to create a new monad $\po((-)^*)$. As we have shown in Theorem \ref{thm:puttingEverythingTogether}, it suffices to build a lifting of monad $\po$ to $\Mon$, the category of algebras for the signature \eqref{syn1} and equations \eqref{law1}. For this purpose we apply the method given in section $\ref{sec:diagres}$.

The first step is lifting $\po$ to the category of $\{\mathtt{skip}, \scp\}$-algebras, which means lifting the operations of $\tt A^*$ to $\po(\tt A^*)$ using a Fubini map. It is well-known that the powerset monad is commutative, and it follows in particular that there exists a unique symmetric monoidal transformation $\te\colon \Pow\times\Pow\to\Pow(-\times-)$ which is given by the Cartesian product: for $U \in \po (X), V\in \po(Y)$, we take  $ \te_{X,Y}(U, V) = U \times V $. Using this Fubini transformation, we can now define the interpretation in $\po (\tt A^*)$ of $\mathtt{skip}$ and $\scp$ as:
\begin{align*}
&\widehat{\mathtt{skip}}=\po (\mathtt{skip}) \circ \eta_1(\ast)=\{\varepsilon\}\\
&\hat{\scp}=\Pow(\scp)\circ \te_{\mathtt{A^*,A^*}}\colon (\Pow \mathtt{A^*})^2 \to \Pow A^*,
\quad (U,V)\mapsto \{u \scp v\mid u\in U, v\in V\}
\end{align*}
To check that this lifting defines a lifting on $\Mon$, we need to check that equations \eqref{law1} hold in $\po( \tt A^*)$. These equations describe associativity and unit: by Theorem \ref{thm:auc}, they are always preserved by a strong commutative  monad like $\po$.

It follows from Theorem \ref{thm:puttingEverythingTogether} and \ref{resi} that we obtain a distributive law $\lambda\colon(\po(-))^* \to \po ((-)^*)$ between monads $(-)^*$ and $\po$, hence the composition $\po ((-)^*)$ is also a monad, allowing us to apply our method again and potentially add another monadic layer. The language $\po (\tt A  ^*)$ contains the lifted versions $\widehat{\tt skip}$ and $\hat{;}$ of our previous constructs  as well as the new operations arising from $\po$, namely a non-deterministic choice operation $+$, which is associative, commutative and idempotent, and its unit $\tt abort$. Note that since the monad structure on $\po((-)^*)$ is defined by a distributive law of $(-)^*$ over $\po$, the set of equations $E$ is made of the equations \eqref{law1} arising from $(-)^*$, the equations \eqref{law2} arising from $\po$, and finally the equations \eqref{law3} expressing distributivity of operations of $(-)^*$ over those of $\po$. The language we have built so far has the structure of an \emph{idempotent semiring}.

\paragraph{\bf Adding the probabilistic layer.}
We will now enrich our language further by adding a probabilistic layer. Specifically, we will add the family of probabilistic choice operators $\oplus_\lambda$ for $\lambda\in [0,1]$ satisfying the axioms of convex algebras, i.e.
\begin{equation}\label{eq:convex}
\prog[p \oplus_\lambda p= p] \qquad \prog[p\oplus_\lambda q=q\oplus_{1-\lambda}p] \qquad
\prog[p\oplus_\lambda(q\oplus_{\tau} r)= (p\oplus_{\hspace{-3pt}\frac{\lambda}{\lambda+(1-\lambda)\tau}}\hspace{-2pt}q)\oplus_{\lambda+(1-\lambda)\tau} r]
\end{equation}
From a monadic perspective, we want to examine the composition of monads $\di(\po((-)^*))$. It is known (see \cite{2003:VaraccaProbability}) that $\di$ does not distribute over $\po$. We will see that our method confirms this result.

\noindent We start by lifting the constants and operations $\{\tt skip, abort, \scp,+\}$ of $\po((-)^*)$ by defining a Fubini map $\te: \di(-)\times \di(-)\to \di(-\times -)$. It is well-known that $\di$ is a commutative monad and that the product of measures defines the Fubini transformation. In the case of finitely supported distributions the product of measures can be expressed simply as follows: given  distributions $\mu\in \di X,\nu\in \di Y$, $\te (\mu, \nu)$ is the distribution on $X\times Y$ defined on singletons $(x,y)\in X\times Y$ by
$
(\te (\mu, \nu))(x,y)=\mu(x)\nu(y) 
$.
Theorem \ref{thm:preservation:noConditions} tells us that associativity, commutativity and unit are preserved by $\di$. It follows that the associativity of both $\scp$ and $\tt +$ is preserved by the lifting operation, and the liftings of $\tt skip$ and $\tt abort$ are their respective units. Furthermore, the lifting of $\tt +$ is commutative. 

We know from Theorem \ref{thm:preservation:relevant} that the idempotency of $\tt +$ will be preserved if $\di$ is relevant. It is easy to see that $\di$ is badly non-relevant: consider the set $X=\{a,b\}, a\neq b$ and any measure $\mu$ on $X$ which assigns non-zero probability to both $a$ and $b$. We have:
\begin{align*}
\te(\Delta_{\di X}(\mu))(a,b)&=(\te (\mu, \mu))(a,b)\\
&=\mu(a)\mu(b)\neq 0 \\
&=\mu(\emptyset)\\ &=\mu\{x\in X\mid \Delta_X(x)=(a,b)\} \\&=\di(\Delta_X)(\mu)(a,b)
\end{align*}
It follows that we \emph{cannot} conclude that the lifting $\widehat{\di}: \Alg[\{\tt skip, abort, \scp,+\}] ~\to\Alg[\{\tt skip, abort, \scp,+\}]$ defined by the product of measures following \eqref{diag:HSoverT} sends idempotent semirings to idempotent semirings, and therefore we cannot conclude that $\di(\po(-)^*)$ is a monad (in fact we know it isn't). It is very telling that idempotency also had to be dropped in the design of the probabilistic network specification language ProbNetKAT (see \cite[Lemma 1]{foster2016probabilistic}) which is very similar to the language we are trying to incrementally build in this Section.

\newcommand{\Mset}{\mathcal{M}}
Requiring that $+$ be idempotent is an algebraic obstacle, so let us now remove it and replace as our inner layer the monad building free idempotent semirings -- that is to say $\po(-)^*$ -- by the monad building free semirings -- that is to say $\Mset(-)^*$, where $\Mset$ is the multiset monad ($\Mset$ can also be described as the free commutative monoid monad). Since we have already checked that the $\di$-liftings of binary operations preserve associativity, units and commutativity, it only remains to check that they preserve the distributivity of $\scp$ over $+$. The equation for distributivity belongs to the syntactic class covered by Theorem \ref{thm:preservation:relevant} since it has the same set of variables on each side (but one of them is duplicated, so we fall outside the scope of Theorems \ref{thm:preservation:noConditions} and \ref{thm:preservation:affine}). 
%
Since we've just shown that $\di$ is not relevant, it follows that we cannot lift the distributivity axioms. So we must weaken our inner layer even further and consider a structure consisting of two monoids, one of which is commutative. Interestingly, the failure of distributivity was also observed in the development of ProbNetKAT (\cite[Lemma 4]{foster2016probabilistic}), and therefore should not come as a surprise. 

Having removed the two distributivity axioms we are left with only the absorption laws to check. In this case the equation has no variable duplication, but has not got the same number of variables on each side of the equation, absorption therefore falls in the scope of Theorem  \ref{thm:preservation:affine}, and we need to check if $\di$ is affine. Since $\di 1\simeq 1$, it is trivial to see that $\eta_1\circ !=\di !$ and hence $\di$ is affine. By Theorem \ref{thm:preservation:affine}, the absorption law is therefore preserved by the probabilistic extension. It follows that the probabilistic layer $\di$ can be composed with the inner layer consisting of the signature $\{\tt abort, skip, \scp,+\}$ and the axioms
\begin{multicols}{2}
\begin{enumerate}[(i)]
\item $\tt p\scp skip=skip\scp p=p$
\item $\tt (p\scp q)\scp r=p\scp (q\scp r)$
\item $\tt p+abort=abort+p=p $
\item $\tt p+q=q+p  $
\item $\tt (p+q)+r=p+(q+r)$
\item $\tt p\scp abort=abort=abort\scp p$
\end{enumerate}
\end{multicols}
\noindent i.e. two monoids, one of them commutative, with the absorption law as the only interaction between the two operations. This structure, combined with the axioms of convex algebras \eqref{eq:convex} and the distributivity axioms
\begin{multicols}{2}
\begin{enumerate}[(Dst i)]
\item $\tt p\scp(q\oplus_\lambda r)=(p\scp q)\oplus_\lambda(p\scp r)$
\item $\tt (q\oplus_\lambda r)\scp p=(q\scp p)\oplus_\lambda(r\scp p)$
\item $\tt p+(q\oplus_\lambda r)=(p+ q)\oplus_\lambda(p+ r)$
\item $\tt (q\oplus_\lambda r)+ p=(q+ p)\oplus_\lambda(r+ p)$
\end{enumerate}
\end{multicols}
\noindent forms the `best approximate language' combining sequential composition, non-deterministic choice and probabilistic choice. Note that the distributive laws above makes good semantic sense, and indeed hold for the semantics of ProbNetKAT. What we have built modularly in this section is essentially the $\ast$-free and test-free fragment of ProbNetKAT.

\section{Discussion and future work.}

We have provided a principled approach to building programming languages by incrementally layering features on the top one another. We believe that our approach is close in spirit to how programming languages are typically constructed, that is to say by an incremental enrichment of the list of features, and to the search for modularity initiated by foundational papers \cite{moggi1991notions} and \cite{liang1996modular}.

Our method has assumed throughout that the monad for the outer layer had to be monoidal/commutative. Our method can in fact be straightforwardly extended to monads satisfying only \eqref{diag:monoidal:unit} and \eqref{diag:monoidal:multiplication}. In practice however, the generality gained in this way is very limited: only a monoidal monad will lift an associative operation with a left and right unit, and given the importance of sequential composition with $\tt{skip}$, the restriction we have placed on our method appears fairly natural and benign.

We must be careful about how layers are composed together: our approach yields distributive interactions between them, but one might want other sorts of interactions. Consider for example the minimal programming language $\tt P^*$ described in Section \ref{sec:intro}, and assume that we now want to add a concurrent composition operation $\parallel$ to this language with the natural axiom $\tt p\parallel skip=p=p\parallel skip$. 
This addition is not as simple as layering described in Section \ref{sec:application}, as the new construct has to interact with the core layer in a whole new way: $\tt skip$ must be the unit of $\parallel$ as well. In such cases our approach is not satisfactory, and two alternative strategies present themselves to us: we can consider `larger' layers, for example the combined theory of sequential composition, $\tt{skip}$ and $\parallel$ described above as a single entity. However, the more complex an inner layer is, the less likely it is that an outer layer with lift it in its entirety. Alternatively, we may want to integrate our technique with Hyland and Power's methods (\cite{hyland2006combining}) and combine some layers with sums and tensors, and others with distributive laws, depending on semantic and algebraic considerations.

A comment about our `approximate language' strategy is also in order. As explained in Section \ref{sec:diagres}, when an equation of the inner layer prevents the existence of a distributive law we choose to remove this equation, i.e. to loosen the inner layer. Another option is in principle possible: we could constrain the outer layer until it becomes compatible with the inner layer. We would obtain in this case a replacement candidate for one of our monads in order to achieve composition. In the case of $\di(\po(-)^*)$ this
would be a particularly unproductive idea since the only elements of  $\di(\po(-)^*)$ which satisfy the residual diagram for idempotency are Dirac deltas, i.e. we would get back the language $\po(-)^*$.

Another obvious avenue of research is to extend our method to programming languages specified by more than just equations. One example is the so-called `exchange law' in concurrency theory given by $\tt (p\parallel r)\scp(q\parallel s)\sqsubseteq(p\scp q)\parallel(r\scp s)$   which involves a native pre-ordering on the collection of programs, i.e. moving from the category of sets to the category of posets. Another example are Kozen's quasi-equations (\cite{1991:KozenCompleteness}) axiomatizing the Kleene star operations, for example $\mathtt{ p\scp x}\leq \mathtt{x}\Rightarrow \mathtt{p^*\scp x}\leq \tt x$. This problem is much more difficult and involves moving away from monads and distributive laws altogether since quasi-varieties are in general not monadic categories.

\newpage
\bibliographystyle{plain}
\bibliography{layersbib}

\newpage

\end{document}